\documentclass[12pt,reqno]{amsart}

\usepackage{amsfonts,amsmath,amssymb,amsthm}
\usepackage{graphicx}
\usepackage{caption}
\usepackage{enumitem}
\usepackage{amsaddr}
\usepackage{mathrsfs}
\usepackage{mathabx}
\usepackage{bbm}
\usepackage{centernot}
\usepackage{mathtools}
\usepackage{graphicx}
\usepackage{multicol}
\usepackage{dcolumn}
\usepackage{wrapfig}
\newtheorem{Theorem}{Theorem}[section]
\newtheorem{Lemma}[Theorem]{Lemma}

\numberwithin{equation}{section}
\usepackage[a4paper,margin=3cm]{geometry}
\begin{document}
\sloppy
\title[On an adjustment to the Fowler equation]
{On an adjustment to the Fowler equation}

\author[B. \'Alvarez-Samaniego]{Borys \'Alvarez-Samaniego}
\address{\vspace{-8mm}N\'ucleo de Investigadores Cient\'{\i}ficos\\
	    Facultad de Ciencias\\
	    Universidad Central del Ecuador (UCE)\\
    	Quito, Ecuador}
\email{balvarez@uce.edu.ec, borys\_yamil@yahoo.com}

\author[W. P. \'Alvarez-Samaniego]{Wilson P. \'Alvarez-Samaniego}
\address{\vspace{-8mm}N\'ucleo de Investigadores Cient\'{\i}ficos\\
        Facultad de Ciencias\\
	    Universidad Central del Ecuador (UCE)\\
    	Quito, Ecuador}
\email{wpalvarez@uce.edu.ec, alvarezwilson@hotmail.com}

\author[K. Lloacana]{Kevin Lloacana-Unda}
\address{\vspace{-8mm}N\'ucleo de Investigadores Cient\'{\i}ficos\\
        Facultad de Ciencias\\
    	Universidad Central del Ecuador (UCE)\\
    	Quito, Ecuador}
\email{kpaul\_rap@hotmail.com}

\date{July 18, 2023.} 

%\thanks{Data sharing not applicable to this article as no datasets were 
%generated or analysed during the current study.}

\begin{abstract}
Following closely the analysis performed by Andrew C. Fowler to derive 
the first canonical equation for nonlinear dune dynamics, but considering 
some appropriate changes of variables, suitable scalings, and by neglecting 
higher order terms, we obtain an adaptation of the aforementioned equation, 
which contains an additional term, to describe dune morphodynamics. 
\end{abstract}
%%%%%%%%%%%%%%%%%%%%%%%%%%%%%%%%%%%%%%%%%%%%%%%%%%%%%%%%%%%%%%%%%%%%%%%%%%%%%%%%%%%%%%%%
%%%%%%%%%%%%%%%%%%%%%%%%%%%%%%%%%%%%%%%%%%%%%%%%%%%%%%%%%%%%%%%%%%%%%%%%%%%%%%%%%%%%%%%%
\maketitle
\section{\label{sec:level1}INTRODUCTION} %codigo de seccion

Dunes are landforms of sand which are formed by the erosive effect of some mobile medium 
above an erodible substrate. 
There are several types of dunes, such as Barchan, transverse, star, parabolic, and 
others (\cite{D,B,C}). 

A. C. Fowler, relatively recently, proposed the first canonical equation, given by 
(\ref{e1})  below, to model nonlinear dune formation (\cite{D,B,C}). Then, 
B. Álvarez-Samaniego, and 
P. Azerad showed, in \cite{A}, the existence of travelling-waves, and local 
well-posedness, in a subspace of $C^{1}_{b}(\mathbb{R})$, of the aforesaid 
equation, which is called the Fowler equation.

At the beginning of Section \ref{sc:deduc}, we mention some previous results that 
can be found in \cite{A,D,B,C}. More specifically, we first write the Exner equation, 
given in (\ref{e2}), which describes the conservation of bed material (\cite{A,C}). 
Next, by taking into account an additional force due to gravity, an approximation for 
the net stress is written  in (\ref{apt}), which was obtained from a suitable analysis 
considering the Orr-Sommerfeld  equation from a Saint-Venant type turbulent model 
(\cite{C}). Also, an adequate estimate for  the bedload transport is exhibited in 
(\ref{qpp}), see \cite{A,D,B}. 

Subsequently, by performing a convenient change of variable, (\ref{cv2}), and a 
proper scaling of the dependent variable, (\ref{resc1}), we obtain a suited 
approximation, (\ref{eaf3}), to  analyze sand dune dynamics. Then, by neglecting 
the terms of order $O\left( \frac{\varepsilon}{\delta^{\frac{1}{3}}}\right)$,  with 
$0 < \varepsilon \ll 1 \ll \delta$, where $ \varepsilon$ 
and $\delta$  are typical physical scales associated with the amplitude, 
and the length of the dune, respectively, we get (\ref{ecf11}). Finally,  by making 
a suitable change of variable,  (\ref{cv5}), and a handy scaling, (\ref{rfb}), an 
adaptation of the Fowler equation is achieved  in (\ref{ecfff}), which is the main 
result of this manuscript. 
\section{DEDUCTION OF THE EQUATION}\label{sc:deduc}
We begin this section by exposing some previous results, given in \cite{A,D,B,C}. 
The Fowler equation for nonlinear dune formation is given by
\begin{equation}\label{e1} 
 \partial_{2} u (x,t) + \partial_{1} \bigg[ \frac{u^{2}}{2} + \int \limits^{+ \infty}_{0} 
 \xi^{- \frac{1}{3}}\partial_{1} u\left( \cdot - \xi, \cdot\right) d\xi - \partial_{1}
 u \bigg](x,t) = 0, 
\end{equation}
for all $(x,t) \in \mathbb{R} \times [0,+\infty)$, where $u$ represents the dune amplitude. 
The derivation of equation (\ref{e1}) follows from the Exner equation, namely,
\begin{equation}\label{e2}
    \partial_{2} u (x,t) + \partial_{1} (q \circ \tau) (x,t) = 0,
\end{equation}
for all $(x,t) \in \mathbb{R} \times [0,+\infty)$, where  $q$ is the bedload transport, 
and $\tau$ is the stress exerted by the fluid on the erodible bed. Moreover, the stress, 
$\tau$, can be approximated by the expression 
\begin{equation}\label{apt}
    \tau (x,t) \approx \hspace{0.1cm} 1 + u (x,t)  + \int\limits^{+ \infty}_{0} 
    \xi^{-\frac{1}{3}}\partial_{1} u \left( x- \xi,t\right) d \xi + u^{2} (x,t)- 
    \partial_{1} u (x,t), 
\end{equation}
for all $(x,t) \in \mathbb{R} \times [0,+\infty)$.
By considering function $q$, proposed by Meyer-Peter, and Müller, a Taylor expansion is 
performed, around the point $1$, up to order $2$, which is represented by the following 
approximation
    \begin{equation}\label{qpp}
        q(\tau) \approx q(1) + q'(1) (\tau -1) + \frac{q''(1)}{2} ( \tau -1)^{2},
    \end{equation}
for all real number $\tau$ belonging to some neighborhood of $1$. 
Let $(x,t) \in \mathbb{R} \times [0,+\infty)$. By using the Chain Rule, we have that 
    \begin{equation} \label{ee5}
    \begin{split}
        \partial_{1} (q \circ  \tau) (x,t)&= q'(\tau (x,t)) \partial_{1} \tau (x,t).\\
    \end{split}
    \end{equation}
Now, by employing (\ref{qpp}), and (\ref{apt}), we get 
\begin{align}
 q'(\tau(x,t))  \approx& \, q'(1) + q''(1)(\tau (x,t) - 1)\label{e3}\\
 \approx& \, q'(1) +q''(1)\bigg( u (x,t) + u^{2}(x,t) - \partial_{1} u(x,t) \nonumber\\
 &+ \int \limits^{+ \infty}_{0} \xi^{-\frac{1}{3}}\partial_{1} u \left( x- \xi,t\right)
 d \xi  \bigg). \nonumber
\end{align}
It follows from (\ref{apt}) that
\begin{align}
  \partial_{1} \tau (x,t)  \approx& \, \partial_{1} u (x,t) +2u(x,t) \partial_{1} u (x,t) 
  - \partial_{1}^{2} u(x,t) \nonumber\\
  &+\partial_{1} \left( \int \limits^{+ \infty}_{0} \xi^{- \frac{1}{3}}\partial_{1} 
  u\left( \cdot - \xi, \cdot \right) d\xi  \right)(x,t). \label{e4}
\end{align}
By using (\ref{e3}), and (\ref{e4}), we obtain that 
\begin{align}
 q'&(\tau (x,t)) \partial_{1} \tau (x,t)\nonumber\\
 \approx& \, \left[ q'(1) +q''(1)(\tau(x,t) -1) \right] \partial_{1} 
 \tau (x,t)\nonumber\\
 =& \, q'(1)\partial_{1} \tau (x,t) + q''(1)(\tau(x,t)-1)\partial_{1} 
 \tau (x,t)  \nonumber\\
 \approx& \, q'(1)\partial_{1} u (x,t)  +q'(1) \bigg[2u (x,t) \partial_{1} 
 u  (x,t) - \partial_{1}^{2} u  (x,t) \nonumber\\
 & + \partial_{1}  \left( \int \limits^{+ \infty}_{0} 
 \xi^{- \frac{1}{3}}\partial_{1} u \left( \cdot - \xi, \cdot \right) d\xi  \right)(x,t) 
 \bigg] + q''(1)(\tau(x,t)  -1) \partial_{1} \tau (x,t). \label{e5}
\end{align}
From ($ \ref{apt}$), and ($\ref{e4}$), we derive that 
\begin{align}
 (\tau&(x,t)-1)\partial_{1} \tau (x,t)\nonumber\\
 \approx& \, \bigg[ u(x,t) + u^{2}(x,t) + \int \limits^{+ \infty}_{0} 
 \xi^{- \frac{1}{3}}\partial_{1} u\left( x- \xi, t\right) d\xi - \partial_{1} 
 u (x,t)  \bigg] \bigg[\partial_{1} u(x,t)\nonumber\\
 & + 2 u(x,t) \partial_{1} u(x,t) - \partial_{1}^{2} u(x,t)
 + \partial_{1} \left( \int \limits^{+ \infty}_{0} \xi^{- \frac{1}{3}}
 \partial_{1} u\left( \cdot - \xi, \cdot \right) d\xi  \right)(x,t)\bigg]\nonumber\\    
 =& \, u(x,t) \partial_{1} u(x,t) + 2 u^{2}(x,t) \partial_{1} u(x,t) - u(x,t) 
 \partial_{1}^{2} u(x,t) \nonumber \\
 &+u(x,t) \partial_{1} \left( \int \limits^{+ \infty}_{0} \xi^{- \frac{1}{3}}
 \partial_{1} u\left( \cdot - \xi, \cdot \right) d\xi  \right)(x,t)+ u^{2}(x,t) 
 \partial_{1} u(x,t)  \nonumber\\
 &+2 u^{3}(x,t) \partial_{1} u(x,t) - u^{2}(x,t) \partial_{1}^{2} u(x,t) + u^{2}(x,t) 
 \partial_{1}\left( \int \limits^{+ \infty}_{0} \xi^{- \frac{1}{3}}
 \partial_{1} u\left( \cdot - \xi, \cdot \right) d\xi  \right)(x,t)\nonumber\\
 &+ \partial_{1} u(x,t) \left( \int \limits^{+ \infty}_{0} \xi^{- \frac{1}{3}}
 \partial_{1} u\left( x - \xi, t\right) d\xi  \right)+ 2 u(x,t) \partial_{1} u(x,t) 
 \left( \int \limits^{+ \infty}_{0} 
 \xi^{- \frac{1}{3}}\partial_{1} u\left( x - \xi, t\right) d\xi  \right)\nonumber\\
 & +\left( \int \limits^{+ \infty}_{0} \xi^{- \frac{1}{3}}\partial_{1} u\left( x- 
 \xi, t\right) d\xi  \right)\cdot \partial_{1} \left( \int \limits^{+ \infty}_{0} 
 \xi^{- \frac{1}{3}}
 \partial_{1} u\left( \cdot - \xi, \cdot\right) d\xi  \right)(x,t)\nonumber\\
 &-\partial_{1}^{2} u(x,t) \left( \int \limits^{+ \infty}_{0} \xi^{- \frac{1}{3}}
 \partial_{1} u\left( x - \xi, t\right) d\xi  \right)- \left( \partial_{1} u(x,t)
 \right)^{2} - 2 u(x,t) \left( \partial_{1} u(x,t)\right)^{2}\nonumber\\
 & + \partial_{1} u(x,t) \partial_{1}^{2} u(x,t)- \partial_{1} u(x,t) \partial_{1}  
 \left( \int \limits^{+\infty}_{0} \xi^{- \frac{1}{3}}
 \partial_{1} u\left( \cdot- \xi, \cdot\right) d\xi  \right)(x,t).\label{dta}
\end{align}
By substituting (\ref{dta}) into (\ref{e5}), it follows from (\ref{e2}), 
and (\ref{ee5}) that
\begin{align*}
  &\partial_{2} u(x,t)+ q'(1)\partial_{1} u (x,t)  +q'(1) \bigg[2u (x,t) 
 \partial_{1} u (x,t)- \partial_{1}^{2} u (x,t) \nonumber\\
 & + \partial_{1} \left( \int \limits^{+ \infty}_{0} 
 \xi^{- \frac{1}{3}}\partial_{1} u\left( \cdot- \xi, \cdot\right) d\xi  \right)(x,t)
 \bigg]+ q''(1) u(x,t) \partial_{1} u(x,t) \nonumber\\
 &+ q''(1) \bigg[ 3 u^{2}(x,t) \partial_{1} 
 u(x,t)- u(x,t) \partial_{1}^{2} u(x,t)- u^{2}(x,t) \partial_{1}^{2} u(x,t) \nonumber\\  
  &+ 
 2 u^{3}(x,t) \partial_{1} u(x,t)+u(x,t) \partial_{1} \left( \int \limits^{+ \infty}_{0} 
 \xi^{- \frac{1}{3}}
 \partial_{1} u\left( \cdot- \xi, \cdot\right) d\xi  \right)(x,t) \nonumber\\
\end{align*}
\begin{align}
 &  + u^{2}(x,t) \partial_{1} \left( \int \limits^{+ \infty}_{0} \xi^{- \frac{1}{3}}
 \partial_{1} u\left( \cdot- \xi, \cdot\right) d\xi  \right)(x,t)+ \partial_{1} u(x,t) 
 \left( \int \limits^{+ \infty}_{0} \xi^{- \frac{1}{3}}
 \partial_{1} u\left( x- \xi, t\right) d\xi  \right)\nonumber\\
 & - \left( \partial_{1} 
 u(x,t)\right)^{2}+ 2 u(x,t) \partial_{1} u(x,t) \left( \int \limits^{+ \infty}_{0} 
 \xi^{- \frac{1}{3}} \partial_{1} u\left( x- \xi, t\right) d\xi  \right)\nonumber\\
 &+ \left( \int \limits^{+ \infty}_{0} \xi^{- \frac{1}{3}}\partial_{1} 
 u\left( x- \xi, t\right) d\xi  \right)\cdot \partial_{1} \left( 
 \int \limits^{+ \infty}_{0} 
 \xi^{- \frac{1}{3}} \partial_{1} u\left( \cdot - \xi, \cdot\right) d\xi  \right)
 (x,t)\nonumber\\
  &-\partial_{1}^{2} u(x,t) \left( \int \limits^{+\infty}_{0} \xi^{- \frac{1}{3}}
 \partial_{1} u\left( x- \xi, t\right) d\xi  \right)- \partial_{1} u(x,t) \partial_{1}  
 \left( \int \limits^{+ \infty}_{0} 
 \xi^{- \frac{1}{3}}\partial_{1} u\left( \cdot- \xi, \cdot\right) d\xi 
 \right)(x,t)\nonumber\\
 &+ \partial_{1} u(x,t) \partial_{1}^{2} u(x,t) - 2 u(x,t) \left( \partial_{1} 
 u(x,t) \right)^{2}\bigg] \approx 0.\label{emdn}   
\end{align}
By employing the Chain Rule, we obtain from (\ref{emdn}) that 
\begin{align}
 &\partial_{2} u(x,t)+ q'(1)\partial_{1} u (x,t) + [2q'(1) +q''(1)] u (x,t) 
 \partial_{1} u (x,t) - q'(1)
 \partial_{1}^{2} u (x,t) \nonumber\\
 & + q'(1) \partial_{1}  \left( \int ^{+ \infty}_{0} 
 \xi^{- \frac{1}{3}}\partial_{1} u\left(\cdot- \xi, \cdot\right) 
 d\xi  \right)(x,t)+ 2 q''(1) u^{3}(x,t) \partial_{1} u(x,t)\nonumber\\
 &  - q''(1)u(x,t)\partial_{1}^{2} u(x,t) + q''(1) \partial_{1} u(x,t) 
 \partial_{1}^{2} u(x,t) -q''(1) 
 \left( \partial_{1} u(x,t)\right)^{2}\nonumber\\
 &+ q''(1) \partial_{1}  \bigg [ u \bigg(\int ^{+ \infty}_{0} 
 \xi^{- \frac{1}{3}}\partial_{1} u\left( \cdot- \xi, \cdot\right) 
 d\xi \bigg)\bigg](x,t) \nonumber\\
 & + q''(1) \partial_{1}  \bigg [u^{2} \bigg( \int ^{+ \infty}_{0} 
 \xi^{- \frac{1}{3}}\partial_{1} u\left( \cdot- \xi, \cdot\right) d\xi 
 \bigg)\bigg](x,t) \nonumber\\
 &- q''(1) \partial_{1}  \bigg [ \partial_{1} u \bigg( \int ^{+ \infty}_{0} 
 \xi^{- \frac{1}{3}}\partial_{1} u\left( \cdot- \xi, \cdot\right) d\xi 
 \bigg)\bigg] (x,t) \nonumber\\
 & + q''(1) \partial_{1}  \bigg [\frac{1}{2} \bigg( \int ^{+ \infty}_{0} 
 \xi^{- \frac{1}{3}}\partial_{1} u\left( \cdot- \xi, \cdot\right) d\xi 
 \bigg)^{2} \bigg](x,t) \nonumber\\
 &+ 3q''(1) u^{2}(x,t) \partial_{1} u(x,t) - q''(1) u^{2}(x,t) \partial_{1}^{2} 
 u(x,t)\nonumber\\
 &- 2 q''(1) u(x,t) \left( \partial_{1} u(x,t) \right)^{2}\approx 0.\label{eaf}   
\end{align}
Taking into account some typical physical scales of the model (see Figure 
\ref{fig:PMo}),  namely, the amplitude, $0 < \varepsilon \ll 1$, and the 
length of the dune,  $1 \ll \delta$, we get $\varepsilon \ll \delta$. 
\begin{figure}[!htb]
\centering
\includegraphics[width=0.45\textwidth]{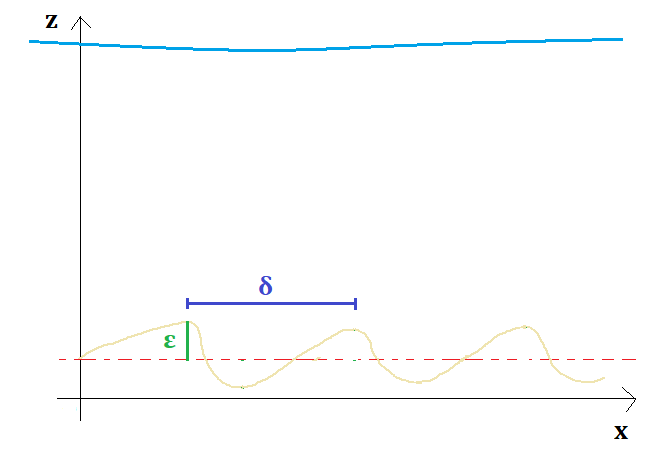}
\caption{Various parameters of the physical model.}
\label{fig:PMo}
\end{figure}
\newpage
\noindent 
We now consider the following function 
\begin{equation}\label{cv2}
 v(y,s) := u(\delta y, \delta s),
\end{equation}
for all $(y,s) \in \mathbb{R} \times [0,+\infty)$. 
Let $(y,s) \in \mathbb{R} \times [0,+\infty)$. Then, 
\begin{align}
 \partial_{1} u (\delta y,\delta s)  
 &= \frac{1}{\delta} \partial_{1}v (y,s), \label{pcv}\\ 
 \partial_{1}^{2} u(\delta y,\delta s)  
 &=\frac{1}{\delta^{2}}\partial_{1}^{2} v  (y,s), \label{pcvo}
\end{align}
and
\begin{equation}\label{pcv2}
 \partial_{2}u(\delta y, \delta s) = \frac{1}{\delta} \partial_{2}v (y,s).
\end{equation}
For all $x \in \mathbb{R}$, we write, $\varphi(x) := \chi_{(0,+\infty)}(x) 
\, x^{-\frac{1}{3}}$, where $ \chi_{(0,+ \infty)}: \mathbb{R} 
\rightarrow \mathbb{R}$  is the characteristic function of the interval 
$(0,+ \infty)$. Thus, 
\begin{align}
 \int \limits^{+ \infty}_{0} \xi^{- \frac{1}{3}}\partial_{1} u\left( 
 \delta y- \xi, \delta s\right) d\xi 
 &= \int \limits^{+ \infty}_{-\infty} \chi_{(0,+\infty)}(\xi) \, 
 \xi^{- \frac{1}{3}}
 \partial_{1} u\left( \delta y- \xi, \delta s\right) d\xi \nonumber\\
 &= \left( \varphi \ast \partial_{1} u(\cdot, \delta s) \right)
 (\delta y)\nonumber\\
 &=: \left( \varphi \ast \partial_{1} u(\cdot, \cdot) 
 \right)(\delta y, \delta s). \label{convT}
\end{align}
It follows from (\ref{pcv}), and (\ref{convT}) that
\begin{align*}
  \partial_{1} &\left( \int \limits^{+ \infty}_{0}  \xi^{- \frac{1}{3}}
 \partial_{1} u\left( \cdot- \xi, \cdot\right) d\xi \right)(\delta y, 
 \delta s) \nonumber\\   
\end{align*}
\begin{align} 
 :=& \, \lim_{h \rightarrow 0} \frac{1}{h} \Bigg[ \Bigg( \int 
 \limits^{+ \infty}_{0} 
 \xi^{- \frac{1}{3}}\partial_{1} u\left( \cdot- \xi, \cdot\right) d\xi \Bigg)
 (\delta y + h, \delta s) \nonumber\\
 &- \Bigg( \int \limits^{+ \infty}_{0}  \xi^{- \frac{1}{3}}\partial_{1} 
 u\left( \cdot- \xi, \cdot\right) d\xi \Bigg)(\delta y, \delta s)\Bigg] \nonumber\\
 =& \, \lim_{h \rightarrow 0} \frac{1}{h} \Bigg[  \int \limits^{+ \infty}_{0}  
 \xi^{- \frac{1}{3}}\partial_{1} u\left( \delta y + h - \xi, \delta s\right) d\xi-  
 \int \limits^{+ \infty}_{0}  \xi^{- \frac{1}{3}}\partial_{1} 
 u\left( \delta y- \xi, \delta s\right) d\xi \Bigg]\nonumber\\
 =& \, \lim_{h \rightarrow 0} \frac{1}{h} \Bigg[  \frac{1}{\delta}\int 
 \limits^{+ \infty}_{0} 
 \xi^{- \frac{1}{3}}\partial_{1} v \bigg( y + \frac{h}{\delta} - 
 \frac{\xi}{\delta},  s \bigg) 
 d\xi  - \frac{1}{\delta} \int \limits^{+ \infty}_{0} 
 \xi^{- \frac{1}{3}}\partial_{1} 
 v \bigg(  y- \frac{\xi}{\delta}, s \bigg) d\xi \Bigg] \nonumber\\
 =& \, \lim_{h \rightarrow 0} \frac{1}{h} \Bigg[  
 \frac{1}{\delta^{\frac{1}{3}}}\int 
 \limits^{+ \infty}_{0}  \zeta^{- \frac{1}{3}}\partial_{1} 
 v\left( y + \frac{h}{\delta} 
 - \zeta,  s\right) d\zeta - \frac{1}{\delta^{\frac{1}{3}}} \int 
 \limits^{+ \infty}_{0}  
 \zeta^{- \frac{1}{3}}
 \partial_{1} v\left(  y- \zeta, s\right) d\zeta \Bigg]\nonumber\\
 =& \, \frac{1}{\delta^{\frac{4}{3}}}\lim_{h \rightarrow 0} 
 \frac{1}{\frac{h}{\delta}} 
 \bigg[  \int \limits^{+ \infty}_{0}  \zeta^{- \frac{1}{3}}
 \partial_{1} v\left( y + 
 \frac{h}{\delta} - \zeta,  s\right) d\zeta-  \int 
 \limits^{+ \infty}_{0}  
 \zeta^{- \frac{1}{3}}\partial_{1} v\left(  y- 
 \zeta, s\right) d\zeta \bigg]\nonumber\\
 =& \, \frac{1}{\delta^{\frac{4}{3}}}\lim_{a \rightarrow 0} \frac{1}{ a} 
 \Bigg[ \int 
 \limits^{+ \infty}_{0}  \zeta^{- \frac{1}{3}}\partial_{1} v\left( y+a - \zeta,  
 s\right) d\zeta-  \int \limits^{+ \infty}_{0}  \zeta^{- \frac{1}{3}}\partial_{1} 
 v\left(  y- \zeta, s\right) d\zeta  \Bigg]\nonumber\\
 =& \, \frac{1}{\delta^{\frac{4}{3}}}\lim_{a \rightarrow 0} \frac{1}{ a} 
 \Bigg[\Bigg( \int\limits^{+ \infty}_{0}  \zeta^{- \frac{1}{3}}\partial_{1} 
 v\left( \cdot - \zeta,  \cdot\right) d\zeta \Bigg)(y+a,s)-  \Bigg(\int 
 \limits^{+ \infty}_{0} \zeta^{- \frac{1}{3}}\partial_{1} 
 v\left(  \cdot- \zeta, \cdot\right) d\zeta \Bigg) (y,s)\Bigg]\nonumber\\
 =:& \, \frac{1}{\delta^{\frac{4}{3}}} \partial_{1}  \Bigg( \int 
 \limits^{+ \infty}_{0}  \zeta^{- \frac{1}{3}}\partial_{1} 
 v\left( \cdot - \zeta,  \cdot\right) d\zeta \Bigg)(y,s)\nonumber\\
 =:& \, \frac{1}{\delta^{\frac{4}{3}}} \partial_{1}  
 (\varphi \ast \partial_{1} v(\cdot, \cdot))(y,s). \label{pcv3}
\end{align}
By using (\ref{pcv}), (\ref{cv2}), and (\ref{pcv3}), we have that
\begin{align*}
 \partial_{1}&  \Bigg [ u  \Bigg(\int \limits^{+ \infty}_{0} 
 \xi^{- \frac{1}{3}}\partial_{1} 
 u\left( \cdot- \xi, \cdot\right) d\xi \Bigg)\Bigg](\delta y,\delta s) \nonumber\\   
 =& \, \partial_{1} u(\delta y,\delta s) \Bigg( \int \limits^{+ \infty}_{0} 
 \xi^{- \frac{1}{3}}\partial_{1} u\left( \delta y- \xi, \delta s\right) d\xi  
 \Bigg) \nonumber\\
 &+ u(\delta y,\delta s) \partial_{1} \Bigg( \int \limits^{+ \infty}_{0} 
 \xi^{- \frac{1}{3}}\partial_{1} u\left( \cdot- \xi, \cdot\right) d\xi  \Bigg)
 (\delta y,\delta s) \nonumber\\   
\end{align*}
\begin{align}
 =& \, \frac{1}{\delta^{2}} \partial_{1} v(y,s) \Bigg( \int \limits^{+ \infty}_{0} 
 \xi^{- \frac{1}{3}}\partial_{1} v\left(  y- \frac{\xi}{\delta},  s\right) d\xi 
 \Bigg) \nonumber\\
 &+\frac{1}{\delta^{\frac{4}{3}}} v(y,s) \partial_{1} \Bigg( \int 
 \limits^{+ \infty}_{0} \zeta^{- \frac{1}{3}}\partial_{1} v \left( \cdot-  
 \zeta,  \cdot\right) d\zeta \Bigg)( y, s) \nonumber\\
 =& \, \frac{1}{\delta^{\frac{4}{3}}} \partial_{1} v(y,s) \Bigg( \int 
 \limits^{+ \infty}_{0} \zeta^{- \frac{1}{3}}\partial_{1} v\left( y- \zeta,  
 s\right) d\zeta  \Bigg) \nonumber\\
 &+ \frac{1}{\delta^{\frac{4}{3}}} v(y,s) \partial_{1} \Bigg( \int 
 \limits^{+ \infty}_{0}\zeta^{- \frac{1}{3}}\partial_{1} v \left( \cdot-  
 \zeta,  \cdot\right) d\zeta \Bigg)( y, s) \nonumber\\
 =& \, \frac{1}{\delta^{\frac{4}{3}}} \partial_{1} \bigg[v \cdot (\varphi \ast 
 \partial_{1} v(\cdot, \cdot)) \bigg](y,s), \label{rdc1}
\end{align}
\begin{align}
  \partial_{1}&  \Bigg [ u^{2} \Bigg(\int \limits^{+ \infty}_{0} 
 \xi^{- \frac{1}{3}}\partial_{1} u\left( \cdot- \xi, \cdot\right) d\xi \Bigg)
 \Bigg](\delta y,\delta s) \nonumber\\ 
 =& \, 2 u(\delta y, \delta s)\partial_{1} u(\delta y,\delta s) 
 \Bigg( \int \limits^{+ \infty}_{0} \xi^{- \frac{1}{3}}\partial_{1} 
 u\left( \delta y- \xi, \delta s\right) d\xi  \Bigg) \nonumber\\
 &+ u^{2}(\delta y,\delta s) \partial_{1} \Bigg( \int \limits^{+ \infty}_{0} 
 \xi^{- \frac{1}{3}}\partial_{1} u\left( \cdot- \xi, \cdot\right) d\xi  \Bigg)
 (\delta y,\delta s) \nonumber\\    
 =& \, \frac{2}{\delta^{2}} v(y,s)\partial_{1} v( y, s) 
 \Bigg( \int \limits^{+ \infty}_{0} \xi^{- \frac{1}{3}}\partial_{1} 
 v\left( y- \frac{\xi}{\delta},  s\right) d\xi  \Bigg) \nonumber\\
 &+ \frac{1}{\delta^{\frac{4}{3}}}v^{2}( y, s)  \partial_{1} \Bigg( \int 
 \limits^{+ \infty}_{0} \zeta^{- \frac{1}{3}}\partial_{1} 
 v\left( \cdot- \zeta, \cdot\right) d\zeta  \Bigg)( y, s) \nonumber\\
  =& \, \frac{1}{\delta^{\frac{4}{3}}} 2 v(y,s)\partial_{1} v( y, s) 
 \Bigg( \int \limits^{+ \infty}_{0} \zeta^{- \frac{1}{3}}\partial_{1} 
 v\left( y- \zeta,  s\right) d\zeta  \Bigg) \nonumber\\
 &+ \frac{1}{\delta^{\frac{4}{3}}}v^{2}( y, s)  \partial_{1} 
 \Bigg( \int \limits^{+ \infty}_{0} \zeta^{- \frac{1}{3}}\partial_{1} 
 v\left( \cdot- \zeta, \cdot\right) d\zeta  \Bigg)( y, s) \nonumber\\
 =& \, \frac{1}{\delta^{\frac{4}{3}}} \partial_{1} \bigg[ v^{2} \cdot 
 (\varphi \ast \partial_{1} v(\cdot,\cdot)) \bigg](y,s), \label{rdc2}    
\end{align}
%\vspace{-0.5 cm}
and
\begin{align}
 \partial_{1}&  \Bigg [\frac{1}{2} \Bigg( \int \limits^{+ \infty}_{0} 
 \xi^{- \frac{1}{3}}\partial_{1} u\left( \cdot- \xi, \cdot\right) 
 d\xi \Bigg)^{2} \Bigg](\delta y,\delta s) \nonumber\\
 =& \,  \Bigg(\int \limits^{+ \infty}_{0} \xi^{- \frac{1}{3}}
 \partial_{1} u\left( \delta y- \xi, \delta s\right) d\xi \Bigg) \cdot 
 \partial_{1} \Bigg( \int \limits^{+ \infty}_{0}  \xi^{- \frac{1}{3}}
 \partial_{1} u\left( \cdot- \xi, \cdot\right) d\xi \Bigg)(\delta y, 
 \delta s) \nonumber\\
 =& \, \frac{1}{\delta^{\frac{7}{3}}} \Bigg( \int \limits^{+ \infty}_{0} 
 \xi^{- \frac{1}{3}}\partial_{1} v\left(  y- \frac{\xi}{\delta},  s\right) 
 d\xi \Bigg) \cdot \partial_{1} \Bigg( \int \limits^{+ \infty}_{0} 
 \zeta^{- \frac{1}{3}} \partial_{1} v\left( \cdot- \zeta, \cdot\right) 
 d\zeta  \Bigg)( y, s) \nonumber\\
 =& \, \frac{1}{\delta^{\frac{5}{3}}} \Bigg(\int \limits^{+ \infty}_{0} 
 \zeta^{- \frac{1}{3}}\partial_{1} v\left(  y- \zeta,  s\right) d\zeta \Bigg) 
 \cdot \partial_{1} \Bigg( \int \limits^{+ \infty}_{0} \zeta^{- \frac{1}{3}}
 \partial_{1} v\left( \cdot- \zeta, \cdot\right) d\zeta  
 \Bigg)( y, s) \nonumber\\ 
 =& \, \frac{1}{\delta^{\frac{5}{3}}} \partial_{1} \bigg[ \frac{1}{2}
 (\varphi \ast \partial_{1} v(\cdot, \cdot ))^{2}\bigg](y,s). \label{rdc3}
\end{align}
%\vspace{-1cm}
Moreover, by employing (\ref{pcvo}), (\ref{pcv}), and (\ref{pcv3}), we get
\begin{align}
  \partial_{1}& \Bigg[ \partial_{1} u \Bigg( \int \limits^{+ \infty}_{0} 
 \xi^{- \frac{1}{3}}\partial_{1} u\left( \cdot- \xi, \cdot\right) d\xi 
 \Bigg)\Bigg] (\delta y,\delta s) \nonumber\\
 =& \, \partial_{1}^{2} u(\delta y, \delta s) \Bigg( \int 
 \limits^{+ \infty}_{0} \xi^{- \frac{1}{3}}\partial_{1} 
 u\left( \delta y- \xi, \delta s\right) d\xi \Bigg) \nonumber\\ 
 &+ \partial_{1} u(\delta y, \delta s) \partial_{1} \Bigg( \int 
 \limits^{+ \infty}_{0}  \xi^{- \frac{1}{3}}\partial_{1} 
 u\left( \cdot- \xi, \cdot\right) d\xi \Bigg)(\delta y, \delta s)\nonumber\\
 =& \, \frac{1}{\delta^{3}} \partial_{1}^{2} v( y, s) \Bigg(\int 
 \limits^{+ \infty}_{0} \xi^{- \frac{1}{3}}\partial_{1} 
 v\left(  y- \frac{\xi}{\delta},  s\right) d\xi \Bigg) \nonumber\\
 &+\frac{1}{\delta^{\frac{7}{3}}} \partial_{1} v( y, s) \partial_{1} 
 \Bigg( \int \limits^{+ \infty}_{0} \zeta^{- \frac{1}{3}}\partial_{1} 
 v\left( \cdot- \zeta, \cdot\right) d\zeta  \Bigg)( y, s) \nonumber\\
  =& \, \frac{1}{\delta^{\frac{7}{3}}} \partial_{1}^{2} v( y, s) 
 \Bigg( \int \limits^{+ \infty}_{0} \zeta^{- \frac{1}{3}}\partial_{1} 
 v\left(  y- \zeta,  s\right) d\zeta \Bigg) \nonumber\\     
 &+\frac{1}{\delta^{\frac{7}{3}}} \partial_{1} v( y, s) \partial_{1} 
 \Bigg( \int \limits^{+ \infty}_{0} \zeta^{- \frac{1}{3}}\partial_{1} 
 v\left( \cdot- \zeta, \cdot\right) d\zeta  \Bigg)( y, s) \nonumber\\
 =& \, \frac{1}{\delta^{\frac{7}{3}}} \partial_{1} \bigg[ \partial_{1} 
 v \cdot (\varphi \ast \partial_{1} v(\cdot, \cdot))\bigg](y,s). \label{rdc4}   
\end{align}
From ($\ref{pcv2}$), ($\ref{pcv}$), (\ref{pcvo}), (\ref{cv2}), 
($\ref{pcv3}$),  and (\ref{rdc1})-(\ref{rdc4}), the approximation (\ref{eaf}) 
may be rewritten as 
\begin{align}
 & \frac{1}{\delta}\partial_{2} v(y,s)+  \frac{q'(1)}{\delta}  \partial_{1} 
 v (y,s)   +  \frac{[2q'(1) +q''(1)]}{\delta} v (y,s) \partial_{1} v (y,s) -  
 \frac{q'(1)}{\delta^{2}}\partial_{1}^{2} v(y,s)\nonumber\\
 &+ \frac{q'(1)}{\delta^{\frac{4}{3}}} \partial_{1} \left( \varphi \ast 
 \partial_{1} v (\cdot,\cdot)\right)(y,s)  + \frac{3q''(1)}{\delta} v^{2}(y,s) 
 \partial_{1} v(y,s) + \frac{2 q''(1)}{\delta} v^{3}(y,s) \partial_{1} v(y,s)
 \nonumber\\
 & - \frac{q''(1)}{\delta^{2}}
 v(y,s) \partial_{1}^{2} v(y,s) - \frac{q''(1)}{\delta^{2}} v^{2}(y,s) 
 \partial_{1}^{2} v(y,s)+ 
 \frac{q''(1)}{\delta^{3}} \partial_{1} v(y,s) \partial_{1}^{2} v(y,s)\nonumber\\
 &- \frac{q''(1)}{\delta^{2}} \left( \partial_{1} v(y,s)\right)^{2} - 
 \frac{2 q''(1)}{\delta^{2}} v(y,s) \left( \partial_{1} v(y,s)\right)^{2}\nonumber\\
 &+ \frac{q''(1)}{\delta^{\frac{4}{3}}} \partial_{1} \bigg [ v \cdot 
 \left( \varphi \ast \partial_{1} v (\cdot,\cdot) \right) \bigg](y,s) + 
 \frac{q''(1)}{\delta^{\frac{4}{3}}} \partial_{1} \bigg [v^{2} \cdot 
 \left( \varphi \ast \partial_{1} v (\cdot,\cdot)\right) \bigg](y,s) \nonumber\\
 &- \frac{q''(1)}{\delta^{\frac{7}{3}}} \partial_{1} \bigg [ \partial_{1}v 
 \cdot \left( \varphi \ast \partial_{1} v (\cdot,\cdot)\right) \bigg](y,s) + 
 \frac{q''(1)}{\delta^{\frac{5}{3}}} \partial_{1} \bigg [\frac{1}{2} 
 \left( \varphi \ast \partial_{1} v (\cdot,\cdot)\right)^{2} \bigg](y,s)
 \approx 0. \label{eaf2}
\end{align}
\noindent
Now, we define the function
\begin{equation}\label{resc1}
 w(z,r) := \frac{1}{\varepsilon} v(z,r),
\end{equation}
for all $ (z,r) \in \mathbb{R} \times [0,+\infty)$. Let 
$(z,r) \in \mathbb{R} \times [0, + \infty)$. It follows from 
the approximation (\ref{eaf2}) that
\begin{align*}
 & \frac{\varepsilon}{\delta}\partial_{2} w(z,r)+  
 \frac{q'(1)\varepsilon}{\delta}  \partial_{1} w (z,r) +  \frac{[2q'(1) +q''(1)]
 \varepsilon^{2}}{\delta} w (z,r) 
 \partial_{1} w (z,r) -  \frac{q'(1)\varepsilon}{\delta^{2}}
 \partial_{1}^{2}w(z,r) \\
 &+ \frac{q'(1)\varepsilon}{\delta^{\frac{4}{3}}} \partial_{1}
 \left( \varphi \ast \partial_{1} w (\cdot,\cdot)\right)(z,r) + \frac{3q''(1) 
 \varepsilon^{3}}{\delta} w^{2}(z,r) 
 \partial_{1} w(z,r) \\ 
 &+ \frac{2 q''(1) \varepsilon^{4}}{\delta} w^{3}(z,r) \partial_{1} w(z,r)  
 - \frac{q''(1) \varepsilon^{2}}{\delta^{2}}w(z,r) \partial^{2}_{1} w(z,r) - 
 \frac{q''(1) \varepsilon^{3}}{\delta^{2}} w^{2}(z,r) \partial_{1}^{2}w(z,r)\\
 &+\frac{q''(1) \varepsilon^{2}}{\delta^{3}} \partial_{1} w(z,r) 
 \partial_{1}^{2}w(z,r)- \frac{q''(1)\varepsilon^{2}}{\delta^{2}} \left( \partial_{1} 
 w(z,r)\right)^{2}  - \frac{2 q''(1) \varepsilon^{3}}{\delta^{2}} w(z,r) 
 \left( \partial_{1} w(z,r)\right)^{2} \\
 &+ \frac{q''(1)\varepsilon^{2}}{\delta^{\frac{4}{3}}} \partial_{1} 
 \bigg [ w \cdot \left( \varphi \ast \partial_{1} w(\cdot,\cdot) \right) 
 \bigg](z,r) + \frac{q''(1) \varepsilon^{3}}{\delta^{\frac{4}{3}}} \partial_{1} 
 \bigg [w^{2} \cdot \left( \varphi \ast \partial_{1} w(\cdot,\cdot) 
 \right) \bigg](z,r)\\
 &- \frac{q''(1) \varepsilon^{2}}{\delta^{\frac{7}{3}}} \partial_{1} 
 \bigg [ \partial_{1} w \cdot \left( \varphi \ast \partial_{1} w(\cdot,\cdot) 
 \right) \bigg](z,r) + \frac{q''(1)\varepsilon^{2}}{\delta^{\frac{5}{3}}}\partial_{1} 
 \bigg [\frac{1}{2}  \left( \varphi \ast \partial_{1} w(\cdot,\cdot) 
 \right)^{2} \bigg](z,r)\approx 0.
\end{align*}
Hence, 
\begin{align*}
 & \partial_{2} w(z,r)+  q'(1)  \partial_{1} w (z,r) +  
 [2q'(1) +q''(1)]\varepsilon w (z,r) \partial_{1} w (z,r) 
 -\frac{q'(1)}{\delta}\partial_{1}^{2} w(z,r)\nonumber\\
 &+  \frac{q'(1)}{\delta^{\frac{1}{3}}} \partial_{1} 
 \left( \varphi \ast \partial_{1} w (\cdot,\cdot)\right)(z,r)+ 
 3q''(1) \varepsilon^{2} 
 w^{2}(z,r) \partial_{1} w(z,r)\nonumber\\
 & + 2 q''(1) 
 \varepsilon^{3} w^{3}(z,r) \partial_{1} w(z,r) - \frac{q''(1) 
 \varepsilon}{\delta}w(z,r) \partial_{1}^{2} w(z,r) 
 - \frac{q''(1) \varepsilon^{2}}{\delta} w^{2}(z,r) \partial_{1}^{2} 
 w(z,r)\nonumber \\    
\end{align*}
\begin{align}
 &+ \frac{q''(1) \varepsilon}{\delta^{2}} \partial_{1} w(z,r) 
 \partial_{1}^{2} w(z,r) - \frac{q''(1)\varepsilon}{\delta} 
 \left( \partial_{1} w(z,r)\right)^{2}- \frac{2 q''(1) 
 \varepsilon^{2}}{\delta} w(z,r) 
 \left( \partial_{1} w(z,r)\right)^{2} \nonumber \\
 &+  \frac{q''(1)\varepsilon}{\delta^{\frac{1}{3}}} \partial_{1} 
 \bigg [ w \cdot \left( \varphi \ast \partial_{1} w (\cdot,\cdot)\right) 
 \bigg](z,r) +  \frac{q''(1) \varepsilon^{2}}{\delta^{\frac{1}{3}}} 
 \partial_{1}\bigg [w^{2}\cdot \left( \varphi \ast \partial_{1} w (\cdot,\cdot) 
 \right) \bigg](z,r) \nonumber\\
 &- \frac{q''(1) \varepsilon}{\delta^{\frac{4}{3}}} \partial_{1} 
 \bigg [ \partial_{1} w \cdot \left( \varphi \ast \partial_{1} 
 w (\cdot,\cdot) \right) \bigg](z,r)\nonumber\\
 &+  \frac{q''(1)\varepsilon}{\delta^{\frac{2}{3}}} \partial_{1} 
 \bigg [\frac{1}{2}\left( \varphi \ast \partial_{1} 
 w (\cdot,\cdot)\right)^{2} \bigg](z,r)\approx 0. \label{eaf3}
\end{align}
In what follows, we neglect the terms of the form 
$\displaystyle O\left(\frac{\varepsilon}{\delta^{\frac{1}{3}}}\right)$, as 
$\varepsilon \rightarrow 0^{+}$, and $\delta \rightarrow +\infty$, 
more precisely, we remove the terms containing the following factors 
\begin{equation}\label{pdp}
\varepsilon^{2}, \hspace{0.25cm}  \varepsilon^{3}, \hspace{0.25cm} 
\frac{\varepsilon}{\delta}, \hspace{0.25cm} \frac{\varepsilon^{2}}{\delta}, 
\hspace{0.25cm} \frac{\varepsilon}{\delta^{2}} \hspace{0.25cm} 
\frac{\varepsilon^{2}}{\delta^{\frac{1}{3}}}, \hspace{0.25cm} 
\frac{\varepsilon}{\delta^{\frac{4}{3}}}, \hspace{0.25cm}and
\hspace{0.25cm} \frac{\varepsilon}{\delta^{\frac{2}{3}}},
\end{equation}
when $\varepsilon \rightarrow 0^{+}$, and $\delta \rightarrow +\infty$. 
Therefore,
\begin{align}
 &\partial_{2} w(z,r)+  q'(1)  \partial_{1} w (z,r) -  
 \frac{q'(1)}{\delta}\partial_{1}^{2} w(z,r) + [2q'(1) +q''(1)]
 \varepsilon w (z,r)\partial_{1} w (z,r) 
 \nonumber\\
 &+ \frac{q'(1)}{\delta^{\frac{1}{3}}} \partial_{1} 
 \left( \varphi \ast \partial_{1} w (\cdot, \cdot)\right)(z,r)  + 
 q''(1)\frac{\varepsilon}{\delta^{\frac{1}{3}}} \partial_{1} 
 \bigg [ w \cdot \left( \varphi \ast \partial_{1} w (\cdot, \cdot)
 \right) \bigg](z,r)  \approx 0. \label{ecf11}
\end{align}
\noindent
The bedload transport function, proposed by Meyer-Peter, 
and Müller (\cite{D,B,C}), is given by
\begin{equation}\label{Fmpm}
 q(\tau) = C \, ([\tau - \tau_{c} ]_{+})^{\frac{3}{2}},
\end{equation}
for all $ \tau \in \mathbb{R}$,  where 
$C > 0$, $[x]_{+} = \max \lbrace x, 0 \rbrace$, for all $x \in \mathbb{R}$, 
and $\tau_{c} = 0.047$ is a yield stress, called the \textit{Shields stress}.
\begin{Lemma}\label{mpme}
 Let $ q: \mathbb{R} \rightarrow \mathbb{R}$ be the function given 
 by (\ref{Fmpm}). Then, $ q'(1) \not = 0$.
\end{Lemma}
\begin{proof}
 We see that 
 \begin{align*}
  q'_{+}(\tau_{c}) &= \lim_{h \rightarrow 0^{+}}\frac{q(\tau_{c} + h) 
  - q(\tau_{c})}{h} = \lim_{h \rightarrow 0^{+}} \frac{C \, h^{\frac{3}{2}}}{h}\\
  &= C \lim_{h \rightarrow 0^{+}}  h^{\frac{1}{2}} =0, 
 \end{align*}
 and 
 \begin{equation*}
  q'_{-}(\tau_{c}) = \lim_{k \rightarrow 0^{-}}\frac{q(\tau_{c} + k) 
  - q(\tau_{c})}{k} = \lim_{k \rightarrow 0^{-}} \frac{0}{k} = 0.  
 \end{equation*}
 So, there exists 
 \begin{equation*}
   q'(\tau_{c}) := \lim_{l \rightarrow 0} \frac{q(\tau_{c} + l) - q(\tau_{c})}{l} =0.
 \end{equation*}
 Therefore, 
 \begin{align}
  q'(\tau) &= \left\lbrace 
  \begin{array}{clll}
   \frac{3}{2} C \,(\tau - \tau_{c})^{\frac{1}{2}}&,& \textup{if} & \tau > \tau_{c},\\
   0&,& \textup{if} & \tau \leq \tau_{c},
  \end{array}
  \right. \nonumber\\
  &=: \frac{3}{2} C \, ([\tau - \tau_{c}]_{+})^{\frac{1}{2}}. \label{dq}
 \end{align}
 \end{proof}
\noindent We now consider the function 
\begin{equation}\label{cv5}
 \phi(e,l) := w \left(\frac{e}{\delta} + \frac{l}{\delta}, 
 \frac{l}{q'(1) \delta} \right),
\end{equation} 
for all $(e,l) \in \mathbb{R} \times [0,+\infty)$.
Let $(e,l) \in \mathbb{R} \times [0,+\infty)$. 
By using (\ref{cv5}), and the Chain Rule, we get  
\begin{align}
 \partial_{1} w\left(\frac{e}{\delta} + \frac{l}{\delta}, 
 \frac{l}{q'(1) \delta} \right) & = \delta \, \partial_{1} \phi (e,l), 
 \label{cv12}\\
 \partial_{1}^{2} w\left(\frac{e}{\delta} + \frac{l}{\delta}, 
 \frac{l}{q'(1) \delta} \right) &= \delta^{2} \, \partial_{1}^{2} 
 \phi(e,l),\label{cv10}
\end{align}
and
\begin{equation}\label{cv11}
 \partial_{2} w\left(\frac{e}{\delta} + \frac{l}{\delta}, 
 \frac{l}{q'(1) \delta} \right) = \, q'(1) \,\delta \,\partial_{2} 
 \phi (e,l)- q'(1)\, \delta \,\partial_{1}\phi(e,l).
\end{equation}
From (\ref{cv12}), we obtain that 
\begin{align}
 \left( \varphi \ast \partial_{1} w (\cdot, \cdot)\right)
 \left(\frac{e}{\delta} + \frac{l}{\delta}, \frac{l}{q'(1) \delta} 
 \right) &:=\Bigg(\int\limits_{0}^{+\infty} \zeta^{-\frac{1}{3}} 
 \partial_{1} w(\cdot - \zeta, \cdot) d\zeta\Bigg)\left(\frac{e}{\delta}
 + \frac{l}{\delta}, \frac{l}{q'(1) \delta} \right)\nonumber\\
 &=\int\limits_{0}^{+\infty} \zeta^{-\frac{1}{3}} \partial_{1} 
 w\bigg(\frac{e}{\delta} + \frac{l}{\delta} - \zeta, \frac{l}{q'(1) 
 \delta}\bigg) d\zeta \nonumber\\
 &=\delta\int\limits_{0}^{+\infty} \zeta^{-\frac{1}{3}} \partial_{1} 
 \phi(e - \delta\zeta, l) d\zeta\nonumber\\
 &= \delta^{\frac{1}{3}} \int\limits_{0}^{+\infty} \mu^{-\frac{1}{3}} 
 \partial_{1} \phi(e - \mu, l) d\mu \nonumber\\
 &=: \delta^{\frac{1}{3}} \left( \varphi \ast \partial_{1} \phi (\cdot, 
 \cdot)\right)(e,l), \label{pt3}
\end{align}
and
\begin{align}
 \partial_{1}& \left( \varphi \ast \partial_{1} w (\cdot, \cdot)\right)
 \left(\frac{e}{\delta} + \frac{l}{\delta}, \frac{l}{q'(1) \delta} \right)
 \nonumber\\
 :=& \,\partial_{1} \Bigg(\int\limits_{0}^{+\infty} \zeta^{-\frac{1}{3}} 
 \partial_{1} w(\cdot - \zeta, \cdot) d\zeta\Bigg)\left(\frac{e}{\delta} 
 + \frac{l}{\delta}, \frac{l}{q'(1) \delta} \right)\nonumber\\   
 :=&\, \lim_{h \rightarrow 0} \frac{1}{h}\Bigg[ \int \limits^{+ \infty}_{0}  
 \zeta^{- \frac{1}{3}}\partial_{1} w\left( \frac{e}{\delta} + \frac{l}{\delta}
 +h- \zeta, \frac{l}{q'(1) \delta}\right) d\zeta \nonumber\\
  &- \int \limits^{+ \infty}_{0}  \zeta^{- \frac{1}{3}}\partial_{1} 
 w\left( \frac{e}{\delta} + \frac{l}{\delta}- \zeta, \frac{l}{q'(1) 
 \delta}\right) d\zeta \Bigg]\nonumber\\
 =&\, \lim_{h\rightarrow 0} \frac{1}{h}\Bigg[ \delta \int 
 \limits^{+ \infty}_{0} \zeta^{- \frac{1}{3}}\partial_{1} 
 \phi\left(  e+ \delta h- \delta \zeta, l\right) d\zeta - \delta
 \int \limits^{+ \infty}_{0}  \zeta^{- \frac{1}{3}}\partial_{1} 
 \phi\left( e- \delta\zeta, l\right) d\zeta \Bigg]\nonumber\\   
 =&\, \lim_{h\rightarrow 0} \frac{1}{h}\Bigg[ \delta^{\frac{1}{3}} 
 \int \limits^{+ \infty}_{0} \mu^{- \frac{1}{3}}\partial_{1} 
 \phi\left(  e+ \delta h- \mu, l\right) d\mu - \delta^{\frac{1}{3}}
 \int \limits^{+ \infty}_{0}  
 \mu^{-\frac{1}{3}}\partial_{1} \phi\left( e- \mu, l\right) d\mu 
 \Bigg]\nonumber\\
 =&\, \delta^{\frac{4}{3}} \lim_{h\rightarrow 0} \frac{1}{\delta h}
 \Bigg[ \int \limits^{+ \infty}_{0} \mu^{- \frac{1}{3}}\partial_{1} 
 \phi\left(  e+ \delta h- \mu, l\right) d\mu - \int \limits^{+ \infty}_{0}  
 \mu^{- \frac{1}{3}}\partial_{1} 
 \phi\left( e- \mu, l\right) d\mu \Bigg]\nonumber\\
 =&\, \delta^{\frac{4}{3}} \lim_{k\rightarrow 0} \frac{1}{k}
 \Bigg[ \int \limits^{+ \infty}_{0} \mu^{- \frac{1}{3}}\partial_{1} 
 \phi\left(  e+ k- \mu, l\right) d\mu - \int \limits^{+ \infty}_{0}  
 \mu^{- \frac{1}{3}}\partial_{1} 
 \phi\left( e- \mu, l\right) d\mu \Bigg]\nonumber\\
 =:&\, \delta^{\frac{4}{3}} \partial_{1}\Bigg(\int\limits_{0}^{+\infty} 
 \mu^{- \frac{1}{3}} \partial_{1} \phi (\cdot - \mu,\cdot)d\mu 
 \Bigg)(e,l)\nonumber\\
 =:&\, \delta^{\frac{4}{3}} \partial_{1} (\varphi \ast \partial_{1} 
 \phi(\cdot,\cdot))(e,l).\label{tnl2}
\end{align}
By employing  (\ref{cv12}), (\ref{pt3}), (\ref{cv5}), and (\ref{tnl2}), 
we see that
\begin{align}
 \partial_{1}& \bigg [ w \cdot \left( \varphi \ast \partial_{1} 
 w (\cdot, \cdot)\right) \bigg] \left(\frac{e}{\delta} + 
 \frac{l}{\delta}, \frac{l}{q'(1) \delta} \right)\nonumber\\
 =&\, \partial_{1} w\left(\frac{e}{\delta} + 
 \frac{l}{\delta}, \frac{l}{q'(1) \delta} \right) 
 \left( \varphi \ast \partial_{1} w (\cdot, \cdot)\right) \left(\frac{e}{\delta} 
 + \frac{l}{\delta}, \frac{l}{q'(1) \delta} \right)+ w\left(\frac{e}{\delta} + 
 \frac{l}{\delta}, \frac{l}{q'(1) \delta} \right)\nonumber\\
 &\cdot \partial_{1} \left( \varphi \ast \partial_{1} 
 w (\cdot, \cdot)\right) \left(\frac{e}{\delta} + 
 \frac{l}{\delta}, \frac{l}{q'(1) \delta} \right) \nonumber\\
 =&\, \delta^{\frac{4}{3}} \partial_{1}\phi (e,l) 
 \left( \varphi \ast \partial_{1} \phi (\cdot, \cdot)\right)
 (e,l)+ \delta^{\frac{4}{3}} \phi(e,l) \partial_{1} 
 (\varphi \ast \partial_{1} \phi(\cdot,\cdot))(e,l) \nonumber\\
 =&\, \delta^{\frac{4}{3}} \partial_{1} 
 \bigg [ \phi \cdot \left( \varphi \ast \partial_{1} 
 \phi (\cdot, \cdot)\right) \bigg](e,l). \label{ttn}    
\end{align}
By using (\ref{cv5})-(\ref{cv11}), (\ref{tnl2}), 
and (\ref{ttn}), we get from (\ref{ecf11}) that 
\begin{align*}
 & q'(1)\, \delta\,\partial_{2} \phi(e,l)-q'(1) 
 \,\delta \,\partial_{1} \phi(e,l)+  q'(1)\, \delta \, 
 \partial_{1} \phi(e,l)-  q'(1)\,\delta \,\partial_{1}^{2} \phi(e,l) \\
 &+ [2q'(1) +q''(1)]\,\varepsilon \, \delta \,\phi(e,l) 
 \partial_{1} \phi(e,l)+ q'(1) \,\delta \, \partial_{1} \left( \varphi \ast 
 \partial_{1} \phi (\cdot,\cdot) \right)(e,l) \\
 &+ q''(1) \,\varepsilon \, \delta \partial_{1} \bigg [ \phi\cdot 
 \left( \varphi \ast \partial_{1} \phi (\cdot,\cdot) \right) \bigg](e,l)  
 \approx 0.
\end{align*}
So,
\begin{align}
 & \partial_{2} \phi(e,l) - \partial_{1}^{2} \phi(e,l)+ 
 \frac{[2 q'(1) + q''(1)]\varepsilon}{q'(1)} \phi(e,l) 
 \partial_{1} \phi(e,l)+\partial_{1} \left( \varphi \ast \partial_{1} 
 \phi (\cdot,\cdot)\right)(e,l)\nonumber\\
 &+\frac{q''(1)\varepsilon }{q'(1)} \partial_{1} 
 \bigg [ \phi\cdot \left( \varphi \ast \partial_{1} 
 \phi (\cdot,\cdot)\right) \bigg](e,l)\approx 0. \label{ecf15}
\end{align}
It follows from (\ref{dq}) that 
\begin{equation*}
 \frac{q'(1)}{2q'(1)+q''(1)} = \frac{1.4295}{3(0.953) + 0.75} \approx 0.3961.
\end{equation*}
We now define the function
\begin{equation}\label{rfb}
 \psi(x,t) := \beta \,\phi(x,t),
\end{equation}
for all $(x,t) \in \mathbb{R} \times [0,+\infty)$, where
\begin{align*}
 \beta := \frac{[2 q'(1) + q''(1)]} {q'(1)} \cdot \varepsilon.   
\end{align*} 
From the approximation  (\ref{ecf15}), for all $(x,t)\in 
\mathbb{R}\times [0,+\infty)$,
\begin{equation}\nonumber
 \begin{split}
  &\frac{1}{\beta}\partial_{2}\psi(x,t) - \frac{1}{\beta}\partial_{1}^{2} 
  \psi(x,t)  + \frac{1}{\beta} \psi(x,t) \partial_{1} \psi(x,t) + 
  \frac{1}{\beta}  \partial_{1} \left( \varphi \ast \partial_{1}
  \psi(\cdot,\cdot) \right)(x,t)\\
  &  + \frac{q''(1)\varepsilon}{q'(1) \beta^{2}} \partial_{1} 
  \bigg [\psi \cdot \left( \varphi \ast \partial_{1} 
  \psi (\cdot,\cdot)\right) \bigg](x,t)  \approx 0.
 \end{split}    
\end{equation}
We denote
\begin{equation}
  \eta := \frac{q''(1)}{2q'(1)+q''(1)}.  
\end{equation}
Finally, we obtain the key result of this paper 
\begin{align}
 &\partial_{2}\psi(x,t) + \partial_{1} 
 \bigg[ \frac{\psi^{2}}{2} - \partial_{1}\psi +  
 (\varphi \ast \partial_{1}\psi(\cdot,\cdot)) + \eta \, \psi \cdot 
 \left( \varphi \ast \partial_{1} \psi (\cdot,\cdot)\right) 
 \bigg](x,t) \approx 0,   \label{ecfff}
\end{align} 
\vspace{-0.5cm}
for all $(x,t)\in \mathbb{R}\times [0,+\infty)$.

%%%%%%%%%%%%%%%%%%%%%%%%%%%%%%%%%%%%%%%%%%%%%%%%%%%%%%%%%%%%%%%%%%%%%%%%%%%%%%%%%%%%%%%%
%%%%%%%%%%%%%%%%%%%%%%%%%%%%%%%%%%%%%%%%%%%%%%%%%%%%%%%%%%%%%%%%%%%%%%%%%%%%%%%%%%%%%%%%
\bibliographystyle {unsrt}

\end{document}